\documentclass[11pt, letterpaper]{article} 
\usepackage{amsmath}
\usepackage[hyperindex,breaklinks,colorlinks,citecolor=blue,pagebackref, urlcolor=black]{hyperref}
\usepackage{amsthm,latexsym, amsfonts,amssymb,amstext} 
\usepackage{fullpage}
\usepackage[usenames]{color} 
\usepackage{graphicx}
\usepackage{enumerate}

\makeatletter
\def\saveenum{\xdef\@savedenum{\the\c@enumi\relax}}
\def\resetenum{\global\c@enumi\@savedenum}
\makeatother

\theoremstyle{plain}
\newtheorem{theorem}{Theorem}

\newtheorem{corollary}[theorem]{Corollary}
\newtheorem{lemma}[theorem]{Lemma}

\newtheorem{conjecture}[theorem]{Conjecture}
\newtheorem{problem}[theorem]{Problem}

\theoremstyle{definition}

\theoremstyle{remark}
\newtheorem{remark}[theorem]{Remark}

\newcommand{\ve}{\varepsilon}
\newcommand{\ignore}[1]{{}}

\newcommand{\diam}{{\mathrm{diam}\,}}

\newcommand{\poly}{{\mathrm{poly}}}

\newcommand{\R}{\mathbb{R}}

\DeclareMathOperator{\suc}{succ}

\title{Tight space-noise tradeoffs in computing the ergodic measure}
\author{Mark Braverman\\mbraverm@cs.princeton.edu\\
Princeton University \and  Crist\'obal Rojas \\ crojas@mat-unab.cl \\Universidad Andres Bello  \and Jon Schneider \\ js44@cs.princeton.edu \\
Princeton University \thanks{MB is  supported in part by an NSF
CAREER award (CCF-1149888), a
Turing Centenary Fellowship, a Packard Fellowship in Science and
Engineering, and the Simons Collaboration on Algorithms and Geometry. CR was partially supported by projects Fondecyt 1150222,  DI-782-15/R Universidad Andr\'es Bello and Basal PFB-03 CMM-Universidad de Chile.}}

\begin{document}

\maketitle

\begin{abstract}
In this note we obtain tight bounds on the space-complexity of computing the ergodic measure 
of a low-dimensional discrete-time dynamical system affected by Gaussian noise. If the scale of the noise is $\ve$, and the function describing the evolution of the system is not by itself a 
source of computational complexity, then the density function of the ergodic measure can be approximated within precision $\delta$ in space polynomial in $\log 1/\ve+\log\log 1/\delta$. We also show that 
this bound is tight up to polynomial factors. 

In the course of showing the above, we prove a result of independent interest in space-bounded computation: that it is possible to exponentiate an $n$ by $n$ matrix to an exponentially large power in space polylogarithmic in $n$. 
\end{abstract}



\section{Introduction}

A discrete-time dynamical system is specified by a function $f$ from a space $X$ to itself. One of the most important problems in the study of dynamical systems is to understand the limiting or asymptotic behavior of such systems; in particular, the limiting distribution of the sequence of iterates $x, f(x), f(f(x)), \dots$. Combinations of such distributions give rise to {\em invariant measures} of the system, which describe the asymptotic behavior in statistical terms. The invariant measures are supported on {\em invariant sets}, which provide a topological description instead. Together, these invariant objects completely characterize the asymptotic behavior of the system. 

Ideally, given a dynamical system, we would like to be able to decide properties of its asymptotic behavior or to compute (to within some approximation)  the invariant objects describing it. Unfortunately, in many cases, simple questions regarding this behavior are undecidable \cite{Mo91,AsMalAm95, Wol02,Ka09} and  computing the relevant invariant objects is impossible \cite{BY,BraYam07, GalHoyRoj07c,BBRY}. The general phenomenon behind these results is that, for many classes of dynamical systems, it is possible to `embed' a Turing machine $M$ in the dynamical system so that achieving the algorithmic task we are concerned with is  equivalent to deciding whether $M$ halts. 

In  \cite{BGR}, Braverman, Grigo, and Rojas showed that under the introduction of noise to a dynamical system (for almost all `natural' noise functions), the set of invariant measures becomes computable. Moreover, in many cases, this set is computable efficiently. Specifically, they show (Theorem C in  \cite{BGR}) that if the noise is Gaussian then there is a unique invariant measure $\mu$; moreover, if $f$ is polynomial-time integrable (convolutions of polynomials in $f$ with polynomial functions can be integrated in polynomial-time), then computing this invariant measure to within precision $\delta$ can be done in time $O(\poly(\log 1/\delta))$. 

The purpose of this paper is to investigate the space complexity of computing the invariant measure of a noisy dynamical system. The algorithm given in Theorem C of  \cite{BGR} for computing the invariant measure requires space $O(\poly(\epsilon^{-1}\log\delta^{-1}))$. By applying (and developing) techniques for space-bounded computation, we prove (in Section \ref{sect:ubnd}) the following refinement of Theorem C that runs in space polylogarithmic of that of the original algorithm (albeit at a cost of a quasi-polynomial increase in the running time). 

An additional assumption that we need to make to obtain tight results is that the function $f$ itself is not a source of additional space complexity. We say that $f$ is $S+$log-space integrable, if it is possible to integrate the convolution of powers of $f$ with polynomial functions with precision $\zeta$ in space $O(S+\log\log 1/\zeta)$ (see Section \ref{sec:prelim} for a precise definition)\footnote{In fact, the conclusion of Theorem~\ref{ubnd} follows even if these convolutions 
can be computed in space $\poly(S+\log\log 1/\zeta)$.}.  

\begin{theorem}\label{ubnd}
Let $X=[0,1]$. 
If the noise $p_{f(x)}^{\epsilon}(\cdot)$ is Gaussian, and $f$ is $(\log \frac{1}{\epsilon})+$log-space integrable, then the computation of the invariant measure $\mu$ at precision $\delta$ can be done in space $O\left(\poly\left(\log\frac{1}{\epsilon} + \log\log\frac{1}{\delta}\right)\right)$. 
\end{theorem}

We can also replace the assumption that $f$ is $(\log \frac{1}{\epsilon})+$log-space integrable with the assumption that $f$ is both log-space computable (i.e. that its values can be computed to within error $\zeta$ in space $O(\log \log 1/\zeta)$) and analytic with bounded Taylor series coefficients. In particular, we show that 

\begin{theorem}\label{ubndalt} Let $X=[0,1]$. 
If the noise $p_{f(x)}^{\epsilon}(\cdot)$ is Gaussian, and $f$ is log-space computable, smooth, and (for some $\eta>0$) satisfies $|\partial^{k}f(x)| \leq k!\eta^{k}$ for all $x$, then the computation of the invariant measure $\mu$ at precision $\delta$ can be done in space $O\left(\poly\left(\log \eta + \log\frac{1}{\epsilon} + \log\log\frac{1}{\delta}\right)\right)$. 
\end{theorem}

For the sake of simplicity, in this note we focus on the case where $X = [0,1]$ (as in \cite{BGR}). Both Theorems \ref{ubnd} and \ref{ubndalt}, however, can be generalized to the case where $X=[0,1]^{d}$. For fixed $d$, the space bounds in Theorems \ref{ubnd} and \ref{ubndalt} remain the same; for variable $d$, the space bounds gain an extra factor of $\poly(d)$. We explain this in further detail in Remark \ref{rem:dimension}.

In order to generalize Theorem C of \cite{BGR} and prove Theorems \ref{ubnd} and \ref{ubndalt}, we require a method to exponentiate $n$ by $n$ matrices up to powers potentially as large as $2^{\poly(n)}$
in  space polylogarithmic in $n$, (the traditional method of iterative squaring only works for powers up to $\poly(n)$). To the best of our knowledge, there is no known existing solution to this problem that operates in polylogarithmic space. In Section \ref{sect:matpow}, we present such a solution based on approximating $M^{E}$ via $p(M)$ for some low degree polynomial $p$ (Theorem \ref{matpowers}). This theorem is arguably the main technical innovation of this paper:

\smallskip

\noindent
{\bf Theorem~\ref{matpowers}.}
{\em Given an $n$ by $n$ matrix $M$ whose entries are given up to precision $2^{-\poly(n)}$ and an integer exponent $E = O(2^{\poly(n)})$, there exists an algorithm that computes $M^{E}$ in space $O(\poly(\log n))$ to within precision $2^{-\poly(n)}$ if $||M^{E}|| \leq 2^{n}$ (and otherwise reports that $||M^{E}|| > 2^{n}$). 
}

\smallskip

Finally, in Section \ref{sect:lbnd}, we prove a corresponding lower bound, showing that this upper bound is tight; the space complexity of computing the invariant measure of such a system cannot be further reduced. 

\begin{theorem}\label{lbnd}
Any algorithm that can compute the invariant measure $\mu$ to within precision $\delta$ of a dynamical system with Gaussian noise kernel $p_{f(x)}^{\epsilon}(\cdot)$ and analytic transition function $f(x)$ (that uniformly satisfies $|\partial^{k}f(x)| \leq k!\eta^k$ for some $\eta = \poly(\epsilon^{-1})$) requires space at least $\Omega\left(\log\frac{1}{\epsilon} + \log\log\frac{1}{\delta}\right)$.
\end{theorem}

These theorems provide evidence for the Space-Bounded Church-Turing thesis (SBCT), introduced by the authors in \cite{BRS}. The SBCT roughly states that a physical system with ``memory'' $M$ is only capable of performing computation in the complexity class $\mathbf{SPACE}(M^{O(1)})$, where memory is a measure of the amount of information the system can preserve from one timestep to the next. For dynamical systems with Gaussian noise of variance $\epsilon$, one can show that $M = O(\log \frac{1}{\epsilon})$; the SBCT thus suggests that such dynamical systems are limited to computations in $\mathbf{SPACE}(\poly\log\frac{1}{\epsilon})$, which is implied by Theorem \ref{lbnd}. See Appendix \ref{sec:sbct} for more details.

\subsection{Open Problems}

In this paper we focus exclusively on the case where the noise is Gaussian. It is straightforward to adapt the proofs in this paper to other choices of noise functions. It remains unclear, however, how the space complexity of computing the invariant measures of $f$ depends precisely on the noise function. More specifically, we would like to be able to answer the following problem.

\begin{problem}
Can we associate with every random perturbation a value $M$ so that computing the invariant measure of a dynamical system with this noise can be done in space $O(\poly(\log M + \log\log 1/\delta))$, and moreover that this is tight: given a random perturbation with value $M$, there is some function $f$ whose invariant measures subject to this random perturbation take space $\Omega(\poly(\log M + \log\log 1/\delta))$?
\end{problem}

For the case when the random perturbation is Gaussian with variance $\epsilon^2$, this note shows that it suffices to take $M = \epsilon^{-1}$ (or, in the $d$-dimensional case, $M = \epsilon^{-d}$).

\subsubsection*{Acknowledgments}
We would like to thank Eric Allender for his advice on space-bounded computation. 

\section{Preliminaries} \label{sec:prelim}

\subsection{Discrete-time dynamical systems}

We begin by giving a brief description of the relevant aspects of the theory of discrete time dynamical systems, largely following the notation of \cite{BGR}. For a complete treatment see for instance \cite{Wal82,Pet83,Man87}.

A \textit{dynamical system} is a metric space $X$ representing the set of possible states along with a map $f: X\rightarrow X$ representing the transitions between states. Given an initial state $x \in X$ of the system, the \textit{trajectory} of $x$ is the sequence $\{x, f(x), f(f(x)), \dots\}$. To avoid certain technical pathologies that can arise, throughout the course of this paper we will assume that $X$ is a compact Lebesgue-measurable subset of $\mathbb{R}^d$ and the function $f$ is continuous.

Given a probability measure $\mu$ over $X$, we can define the pushforward of $\mu$ under $f$ via $(f\mu)(A) = \mu(f^{-1}(A))$ for all events $A \subset X$. A probability measure $\mu$ is \textit{invariant} for the dynamical system if $f\mu = \mu$.

In this note, we focus on the case of dynamical systems with noise. Denote by $P(X)$ the set of Borel probability measures over $X$ under the weak convergence topology. A \textit{random perturbation $\mathcal{S}$ of $f$} is given by a family $\{Q_{x}\}_{x\in X} \in P(X)$ of probability measures over $X$ for each point in $x$ which each represent the `noise' at that point. Then, instead of a deterministic trajectory, $\mathcal{S}$ induces a Markov chain over $X$, where $\mathrm{Pr}[x_{t+1} \in A] = Q_{f(x_{t})}(A)$ for all Borel sets $A \subset X$. Likewise, the pushforward of a probability measure $\mu \in P(X)$ under $\mathcal{S}$ is defined by $(\mathcal{S}\mu)(A) = \int_{X}Q_{f(x)}(A)d\mu$. As before, $\mu$ is an \textit{invariant measure} of the random perturbation $\mathcal{S}$ of $f$ if $\mathcal{S}\mu = \mu$. 

For simplicity, throughout this paper we will assume that the domain $X$ is the $d$-dimensional cube $[0,1]^d$ (and for the majority of the discussion, we will focus on the case where $d$ equals $1$). Moreover, in all of our examples we will be concerned with the case of Gaussian noise with variance $2\epsilon^2$, where the measure $Q_{x}$ is defined (in the case $d=1$) by the probability density function

\begin{equation*}
K_{\epsilon}(y, x) = C_{\epsilon}(x)\frac{1}{\epsilon\sqrt{2\pi}}\exp(-(y-x)^2/2\epsilon^2)
\end{equation*}

\noindent
where $C_{\epsilon}(x)$ is a normalization factor so that $K_{\epsilon}(y,x)$ has measure 1 over $[0,1]$; specifically, $C_{\epsilon}(x)$ is given by

\begin{equation*}
C_{\epsilon}(x) = \left(\int_{0}^{1}\frac{1}{\epsilon\sqrt{2\pi}}\exp(-(y-x)^2/2\epsilon^2)dy\right)^{-1}
\end{equation*}

\noindent
Note that if $\mu(x)$ is the density function of a probability measure on $[0,1]$, then the density $\rho = \mathcal{S}\mu$ of the pushforward measure under $\mathcal{S}$ is given by

\begin{equation*}
\rho(x) = \int_{0}^{1} \mu(y) K_{\epsilon}(f(y), x) dy
\end{equation*}

For this reason (following the notation of \cite{BGR}), we will write $K_{f}(y, x)$ as shorthand for $K_{\epsilon}(f(y), x)$.  We will also write $p_{f(x)}^{\epsilon}$ to denote the family $Q_{f(x)}$ of probability measures for this dynamical system with noise (i.e. the probability measure induced by $K_{\epsilon}(y,f(x))$). 

\subsection{Space-bounded computation}\label{sect:spacebound}

The space complexity classes we consider in this paper are very small; they are (poly)logarithmic in the size of the output. To this end, we review some classic results from space-bounded computation. 


A function $f$ is \textit{log-space computable} if it can be computed by a Turing machine with a read-only input tape, a one-way write-only output tape, and a read-write work tape of size $O(\log n)$. The following functions are known to be log-space computable:

\begin{enumerate}[(a)]

\item The composition of a constant number of log-space functions. The composition of two functions $f(g(x))$ can be performed by dividing the work tape into two tapes of size $O(\log n)$, and using the second tape to compute the desired bit of $g(x)$ whenever it is required for $f(g(x))$. By induction, this can be extended to any constant-depth composition of log-space functions. \label{comp}

\item Addition of $\poly(n)$ $n$-bit integers. This can be done with via the standard grade-school addition algorithm (with some attention paid to how to represent carries). \label{add}

\item Multiplication of two $n$-bit integers. This follows from \ref{add} via the standard algorithm for long multiplication. \label{mult}

\item Multiplication of two $n$ by $n$ matrices, each of whose entries is an $n$-bit integer. This follows from \ref{add} and \ref{mult} (each entry is the sum of $n$ products of two $n$-bit integers). \label{matrixmult}

\item Division of two $n$-bit integers. This result is due to Chiu, Davida, and Litow \cite{CDL95}. The main idea of their proof is to represent both numbers in terms of their values modulo various small primes, perform the arithmetic operations modulo these small primes, and reconstruct the result via the Chinese Remainder Theorem. \label{div}

\item Multiplication of $\poly(n)$ $n$-bit integers. This can be done via the same technique of Chinese Remainder representation described in \ref{div} and is also described in \cite{CDL95}.\label{itmult}

\item Arithmetic operations on real numbers up to precision $2^{-\poly(n)}$. This follows from the preceding results (we need only additionally keep track of the location of the decimal/binary point, which requires at most a logarithmic amount of extra space). \label{realarith}

\item Computation of factorials and binomial coefficients. This follows from \ref{div} and \ref{itmult}. \label{facts}

\item Taking products, powers (with exponents of size $\poly(n)$), and compositions of polynomials with degree $\poly(n)$ and coefficients of size $\poly(n)$. This follows from \ref{facts}, \ref{itmult}, and \ref{add}. \label{polys}

\item Computing $\exp$, $\log$, and $\arctan$ of numbers to within precision $2^{-\poly(n)}$. This was originally shown by Alt in \cite{Alt84} (in all cases it suffices to approximate these functions via some sufficiently long prefix of their Taylor series). \label{funcs}

\item Computing $x^{E}$ to within precision $2^{-\poly(n)}$, where $x$ is a real number provided to precision $2^{-\poly(n)}$ and $E$ is a $\poly(n)$-bit integer. Again, this was shown by Alt in \cite{Alt84} and essentially follows from \ref{funcs} by writing $x^{E} = \exp(E\log x)$. For completeness, we provide a derivation of this fact in Appendix \ref{sect:numpow}.

\saveenum
\end{enumerate}

There are some operations which, while we do not know how to perform in a logarithmic amount of space, we do know how to perform in a polylogarithmic amount of space. These include:

\begin{enumerate}[(a)]
\resetenum

\item
Computing the composition of logarithmically many log-space functions. By similar logic as \ref{comp} above, this can be done in space $O(\log^2 n)$.

\item \label{matpow}
Computing $M^{\poly(n)}$ to within precision $2^{-\poly(n)}$, where and $M$ is an $n$-by-$n$ matrix of $\poly(n)$-bit entries. This can be done in space $O(\log^2 n)$ via repeated squaring (this is essentially the logic behind Savitch's theorem, see \cite{Sav70}).

\item \label{det}
Computing the determinant (and more generally, the coefficients of the characteristic polynomial) of an $n$-by-$n$ matrix $M$ with $\poly(n)$-bit integer entries. This can be done in space $O(\log^2 n)$ via a result of Buntrock, Damm, Hertrampf, and Meinel (see \cite{BDMH92}).

\item
Inverting an $n$-by-$n$ matrix $M$ with $\poly(n)$-bit integer entries. This follows from \ref{det} by expressing the inverse of $M$ in terms of the determinant of $M$ and cofactor matrix of $M$. 

\item
Computing all roots of a polynomial of degree $n$ with $\poly(n)$-bit integer coefficients to within precision $2^{-\poly(n)}$. This follows from a result of Neff and Reif; their algorithm uses space $O(\log^7 n)$ (see \cite{NR96}). \label{roots}

\item
Computing the eigenvalues of an $n$-by-$n$ matrix $M$ with $\poly(n)$-bit integer entries to within precision $2^{-\poly(n)}$. This follows from \ref{roots} and \ref{det} by computing the roots of the characteristic polynomial of $M$. \label{eigenvals}

\end{enumerate}

It should be noted that many of these operations, when restricted to polylogarithmic space, require (to the best of our knowledge) superpolynomial running times. In particular, the above algorithm for matrix exponentiation (and more generally, Savitch's algorithm for STCONN) requires time $O(2^{\log^2 n})$. It is open whether every function computable in polylogarithmic space can be computed simultaneously in polylogarithmic space and polynomial time. We therefore cannot ensure the same time bound as in the original statement of Theorem C in \cite{BGR}.

\subsection{Real computation}

Throughout the rest of the paper (and particularly in the next section) we will often have to work with binary representations of real numbers. We summarize in this section some common notation we use in the remainder of the paper.

A real number $x$ is \textit{given up to precision $2^{-n^c}$} if $x$ is given as an integer multiple of $2^{-n^c}$. We say $x$ is \textit{given up to precision $2^{-\poly(n)}$} if it is given up to precision $2^{-n^c}$ for some $c$. We further assume that all numbers given this way are also bounded above in magnitude by $2^{\poly(n)}$. 

We say we can compute a function $f(x)$ up to precision $\delta$ if there is an algorithm which, when provided with $x$ up to a sufficiently high precision, computes a dyadic number $x'$ such that $|x'-f(x)| \leq \delta$. We say we can compute a function $f(x)$ up to precision $2^{-\poly(n)}$ in polylogarithmic (alternatively, logarithmic) space if, for each positive integer $c$, we can compute $f(x)$ up to precision $2^{-n^{c}}$ in space $O(\poly(\log n))$ (alternatively, $O(\log n)$), where the degree of the polynomial is independent of $c$.

In the statement of Theorem \ref{ubnd}, we require that the function $f(x)$ is $(\log \epsilon^{-1})+\log$-space integrable. Formally, a function $f:\R\rightarrow\R$ is \textit{$S+\log$-space integrable}, if, given an interval $[a,b]$ (with $a$ and $b$ both given up to precision $2^{-\poly(n)}$) and a polynomial $p(x)$ of degree $\poly(n)$ whose coefficients are all given to precision $2^{-\poly(n)}$, it is possible to compute the integral $\int_{a}^{b} f(x)p(x)dx$ to within precision $2^{-\poly(n)}$ in space $O(S + \log n)$. In the higher dimensional case where $f$ is a function from $\R^{d}$ to $\R^{d}$, the interval $[a,b]$ is replaced by the box $[a_1, b_1]\times \dots \times [a_d, b_d]$. 

Finally, we define what we mean by the computation of an invariant measure of a dynamical system. We say a measure $\mu'$ agrees with a measure $\mu$ up to precision $\delta$ if the total variation distance between $\mu$ and $\mu'$ is at most $\delta$. If measures $\mu$ and $\mu'$ are given by density functions, we will write $||\mu - \mu'||_{\infty}$ to denote the $L_{\infty}$ distance between the two density functions; note that since the size of our domain is normalized to $1$, if $||\mu - \mu'||_{\infty} \leq \delta$, then the total variation distance between $\mu$ and $\mu'$ is also at most $\delta$. We say we can compute a measure $\mu$ in space $O(S)$ if, for any interval $[a, b]$, we can approximate the weight of $\mu$ over $[a,b]$ to within precision $2^{-n}$ in space $O(S + \log n)$ (again, in the $d$-dimensional case, we replace the interval $[a,b]$ with the box $[a_1, b_1] \times \dots \times [a_d, b_d]$).  

\section{Exponentiating matrices to large powers} \label{sect:matpow}

Iterative squaring allows us to compute powers of $n$-bit matrices up to exponents that are polynomial in $n$ in polylogarithmic space. Proving Theorem \ref{ubnd}, however, requires us to be able to exponentiate numbers and matrices up to exponents of size potentially exponential in $n$.

In this section we demonstrate how to raise matrices to exponentially large exponents using a polylogarithmic amount of space. In particular, we prove the following theorem. 

(Throughout this section, we take the norm $||M||$ of a matrix $M$ to be the maximum norm, i.e. the maximum absolute value of an entry of $M$). 

\begin{theorem}\label{matpowers}
Given an $n$ by $n$ matrix $M$ whose entries are given up to precision $2^{-\poly(n)}$ and an integer exponent $E = O(2^{\poly(n)})$, there exists an algorithm that computes $M^{E}$ in space $O(\poly(\log n))$ to within precision $2^{-\poly(n)}$ if $||M^{E}|| \leq 2^{n}$ (and otherwise reports that $||M^{E}|| > 2^{n}$). 
\end{theorem}

Our general approach will be to construct a polynomial $p(x)$ of degree at most $n$ such that, for each eigenvalue $\lambda$ of $M$, $p(\lambda) \approx \lambda^{E}$.  It will then follow that $p(M) \approx M^{E}$. 

We first show that we can reduce Theorem \ref{matpowers} to the case where $M$ is diagonalizable with $n$ distinct eigenvalues. 

\begin{theorem}\label{nonsingdiag}
Given any $n$ by $n$ matrix $M$ whose entries are given up to precision $2^{-\poly(n)}$, an integer exponent $E \leq 2^{\poly(n)}$ (that satisfies $||M^{E}|| \leq 2^n$) and a precision $\delta = \Omega(2^{-\poly(n)})$, there exists an algorithm that computes in space $O(\poly(\log n))$ a matrix $M_0$ with entries provided to precision $2^{-\poly(n)}$ such that $M_0$ has $n$ distinct eigenvalues and $||M^{E} - M_0^{E}|| \leq \delta$. 
\end{theorem}
\begin{proof}
Let $D$ be the diagonal matrix $\mathrm{diag}(1, 2, 3, \dots, n)$, and set 

\begin{equation*}
M(t) = M(1-t) + Dt
\end{equation*}

Let $p(t)$ be the discriminant of the characteristic polynomial of the matrix $M(t)$; that is, if $\lambda_i(t)$ are the roots of the characteristic polynomial of $M(t)$, then

\begin{equation}\label{eq:disc}
p(t) = \prod_{i<j} (\lambda_i(t)-\lambda_j(t))^2
\end{equation}

It is known that the discriminant of a polynomial $P(x)$ of degree $d$ can be computed as the determinant of a $(2d-1)$ by $(2d-1)$ matrix whose entries are coefficients of $P(x)$ (see for instance \cite{GKZ94}). Since the coefficients of the characteristic polynomial matrix are in turn polynomials in the entries of $M$, it follows that $p(t)$ is a polynomial in $t$. Moreover, by equation \ref{eq:disc}, scaling a matrix by some multiplicative factor $c$ multiplies the discriminant of the characteristic polynomial of this matrix by a factor of $c^{n(n-1)}$; it follows that the discriminant of the characteristic polynomial of a matrix is a homogeneous polynomial of degree $n(n-1)$ in the entries of the matrix, and therefore $p(t)$ has degree at most $n(n-1)$. Finally, since we can compute determinants and characteristic polynomials of matrices in polylogarithmic space (by remark \ref{det} in Section \ref{sect:spacebound}), we can compute $p(t)$ in polylogarithmic space.

Note that since $M(1) = D$, it follows that $p(1) = \prod_{i<j} (i-j)^2 \neq 0$, and therefore that $p(t)$ is not identically $0$. Now, let $t_0$ be the largest power of $2$ satisfying

\begin{equation*}
t_{0} = 2^{-e_0}  \leq \frac{\delta}{100 n(n-1) 2^n E^2 ||D - M||}
\end{equation*}

\noindent
and consider the $n(n-1)+1$ values $t = kt_0$ where $k$ ranges from $0$ to $n(n-1)$ inclusive. Since $p(t)$ is a polynomial of degree $n(n-1)$ that is not identically $0$, it can have at most $n(n-1)$ roots, so for at least one of these choices of $k$, $p(t) \neq 0$. Since $p(t)$ is non-zero, no two eigenvalues of $M(t)$ are equal. On the other hand, for this value of $t$, note that

\begin{eqnarray*}
||M^{E} - M(t)^{E}|| &\leq & \left|\left| M^{E} - \left(M + (D-M)\frac{\delta k}{100 n(n-1) 2^{n} E^2 ||D - M||} \right)^{E}\right|\right| \\
&\approx & \left|\left| \frac{\delta k (D-M)}{100 n(n-1)2^{n} ||D-M||} \right| \right|  ||M^{E-1}|| \\
& \leq & \dfrac{\delta}{100}
\end{eqnarray*}

\noindent
It therefore suffices to take $M_0 = M(t)$. Since $t_0 = 2^{-\poly(n)}$, the entries of $M_0$ are all given to precision $2^{-\poly(n)}$, as desired.

\end{proof}

We next cite the following technical lemma about the minimum distance between distinct eigenvalues of $M$.

\begin{lemma} \label{rootdist}
Let $p(x)$ be a degree $n$ polynomial whose coefficients are integers all with absolute value at most $A$. Then for any two distinct roots $r_i\neq r_j$ of $p(x)$,

\begin{equation}
|r_i - r_j| \geq 2nA^{-n^2}
\end{equation}
\end{lemma}
\begin{proof}
See \cite{Col01}.
\end{proof}
\begin{corollary}\label{eigenvaldist}
Let $M$ be an $n$ by $n$ matrix whose entries are provided to precision $2^{-\poly(n)}$ and are at most $2^{\poly(n)}$ in absolute value. Then for any two distinct eigenvalues $\lambda_i \neq \lambda_j$ of $M$, $|\lambda_i - \lambda_j| \geq 2^{-\poly(n)}$.
\end{corollary}
\begin{proof}
If the entries of $M$ are provided to within precision $2^{-a(n)}$, consider $2^{a(n)}M$. This is an integer matrix whose entries are all  of size at most $2^{\poly(n)}$. It follows that the coefficients of the characteristic polynomial of this matrix have absolute value at most $2^{\poly(n)}$, and hence (by Lemma \ref{rootdist}),

\begin{equation*}
|2^{a(n)}\lambda_i - 2^{a(n)}\lambda_j| \geq 2n\left(2^{\poly(n)}\right)^{-n^2} = 2^{-\poly(n)}
\end{equation*}

\noindent
and hence

\begin{equation*}
|\lambda_i - \lambda_j| \geq 2^{-\poly(n)}
\end{equation*}
\end{proof}

Finally, we prove the following lemma bounding the size of the matrices related to the eigendecomposition of a matrix $M$.

\begin{lemma}\label{eigdecompbnd}
Let $M$ be an $n$ by $n$ non-singular matrix with distinct eigenvalues whose entries are provided to precision $2^{-\poly(n)}$, and let $D'$ be a diagonal matrix all of whose diagonal entries have absolute value at most $1$. Then if we write $M = U^{-1}DU$, the matrix $M' = U^{-1}D'U$ has entries at most $2^{\poly(n)}$.
\end{lemma}
\begin{proof}
Let $\lambda_1, \lambda_2, \dots, \lambda_n$ be the eigenvalues of $M$ (i.e. the diagonal entries of $D$), and let $\mu_1, \mu_2, \dots \mu_n$ be the diagonal entries of $D'$. Consider the polynomial $p(x)$ of degree at most $n-1$ which maps $\lambda_i$ to $\mu_i$ for each $i$. By the Lagrange interpolation theorem, we can write $p(x)$ as

\begin{equation*}
p(x) = \sum_{i=1}^{n}\prod_{j\neq i}\mu_i\dfrac{(x-\lambda_j)}{(\lambda_i - \lambda_j)}
\end{equation*}

By Corollary \ref{eigenvaldist}, for all $i\neq j$, $|\lambda_i - \lambda_j| \geq 2^{-\poly(n)}$. Combining this with the fact that $|\mu_i| \leq 1$ implies that all coefficients of $p(x)$ are at most $2^{\poly(n)}$ in absolute value.

Consider now the matrix $p(M)$. Note that since $p(D) = D'$, $p(M) = M'$. But since all the entries of $M$ are at most  $2^{\poly(n)}$, the entries of $p(M)$ will be at most $2^{\poly(n)}$, and hence the entries of $M'$ are at most $2^{\poly(n)}$. 

\end{proof}

We now proceed to prove Theorem \ref{matpowers}.

\begin{proof}[Proof of Theorem~\ref{matpowers}]
By Theorem \ref{nonsingdiag} we can assume without loss of generality that $M$ is diagonalizable with distinct eigenvalues. We begin by finding the eigenvalues of $M$. By remark \ref{eigenvals} of Section \ref{sect:spacebound}, it is possible in polylogarithmic space to compute the eigenvalues of $M$ to within any precision $2^{-\poly(n)}$. 

Let $\lambda_1, \lambda_2, \dots, \lambda_n$ be the eigenvalues of $M$.  For each $\lambda_{i}$, let $\tilde{\lambda}_i$ be our approximation to $\lambda_i$ (so that $|\tilde{\lambda}_i - \lambda_i| \leq 2^{-\poly(n)}$ for some choice of $\poly(n)$). We now construct via Lagrange interpolation the polynomial $p(x)$ such that for each $i$, $p(\tilde{\lambda}_i) = \tilde{\lambda_{i}}^{E}$ (note that by Theorem \ref{cpxpowers}, we can compute $\tilde{\lambda_{i}}^{E}$ to within any precision $2^{-\poly(n)}$ in logarithmic space). We wish to show that we can ensure (via approximating the roots with fine enough precision) that $|p(\lambda_{i}) - \lambda_{i}^{E}| \leq 2^{-\poly(n)}$ for any given choice of precision $2^{-\poly(n)}$.

To show this, first note that the Lagrange interpolation formula says that we can write $p(x)$ as

\begin{equation*}
p(x) = \sum_{i=1}^{n}\prod_{j\neq i}\tilde{\lambda}_i^{E}\dfrac{(x-\tilde{\lambda}_j)}{(\tilde{\lambda_i} - \tilde{\lambda_j})}
\end{equation*}

Recall that, by Corollary \ref{eigenvaldist}, for all $i\neq j$, $|\tilde{\lambda}_i - \tilde{\lambda_j}| \geq 2^{-n^{a}}$, for some constant $a$. In addition, $\tilde{\lambda}_i^{E}$ is at most $||M^{E}||$ which by our assumption is at most $2^{n}$. Hence, all the coefficients of $p(x)$ have magnitude at most $2^{n}\left(2^{-n^{a}}\right)^{-n} \leq 2^{n^{a+2}}$. 

Next, note that if $p(x)$ is a polynomial of degree $d$ all of whose coefficients are at most $A$ in absolute value, then

\begin{eqnarray}
|p(x) - p(y)| &\leq & A \sum_{i=0}^{d} |x^{i} - y^{i}| \\
&=&  A|x-y| \sum_{i=1}^{d}\left|\sum_{j=0}^{i-1}x^{j}y^{i-j}\right| \\
&\leq & d^2 A \max(|x|,|y|)^{d-1} |x-y|
\end{eqnarray}

Since $|\tilde{\lambda}_i - \lambda_i| \leq 2^{-n^{b}}$ for some constant $b$ and $|\lambda_i|^{d-1} \leq |\lambda_i|^{E} \leq 2^{n}$, then it follows that,

\begin{equation*}
|p(\lambda_{i}) - p(\tilde{\lambda}_i)| \leq n^2 2^{n^{a+2}-n^{b}+n}
\end{equation*}

Therefore, as long as we choose $b > a+2$, $|p(\lambda_{i}) - p(\tilde{\lambda}_i)|$ will be at most $2^{-O(n^{b})}$. Since $p(\tilde{\lambda}_i) = \tilde{\lambda}_i^{E}$, and since $|\tilde{\lambda}_i^{E} - \lambda_{i}^{E}| \approx E|\tilde{\lambda}_i - \lambda_i| \leq 2^{-n^{b}}|\tilde{\lambda}_i - \lambda_i|$, it follows that 

\begin{equation*}
|p(\lambda_{i}) - \lambda_{i}^{E}| \leq 2^{-O(n^{b})} + E 2^{-n^{b}}
\end{equation*}

Since $E \leq 2^{\poly(n)}$, $E \leq 2^{n^{c}}$ for some $c$. For any $c'$, choosing $b = c+c'$ ensures that $|p(\lambda_i) - \lambda_{i}^{E}| \leq 2^{-n^{c'}}$, as desired.

Finally, consider the matrix $p(M)$. We claim that each entry of $p(M) - M^{E}$ has absolute value at most $2^{-\poly(n)}$. To see this, note that if we diagonalize $M$ as $M = U^{-1}DU$, where $D$ is a diagonal matrix containing the eigenvalues of $M$, then $p(M) - M^{E} = U^{-1}(p(D) - D^{E})U$. Each diagonal entry of $p(D) - D^{E}$ is of the form $p(\lambda_i) - \lambda_i^{E}$ and therefore by the above discussion has magnitude at most $2^{-n^{c'}}$, for any $c'$ of our choosing. Rewriting $p(M) - M^{E}$ in the form $2^{-n^{c'}}U^{-1}2^{n^{c'}}(p(D)-D^{E})U$ and applying Lemma \ref{eigdecompbnd}, it follows that (for sufficiently large $c'$) each entry of $p(M) - M^{E}$ also has magnitude $2^{-\poly(n)}$.

It therefore suffices to compute $p(M)$. Since we can compute the coefficients of the polynomial $p$ in polylogarithmic space and since we can compute $M^{k}$ for any $k \leq n$ in polylogarithmic space via repeated squaring, we can compute $p(M)$ in polylogarithmic space, as desired.
\end{proof}

\section{Computing invariant measures in small space}\label{sect:ubnd}

In this section we prove Theorem \ref{ubnd}. 

This theorem can be seen as a refinement of Theorem C in  \cite{BGR}. Our strategy, therefore, will be mainly to adapt the algorithm described in the proof of Theorem C, taking care to implement each step in polylogarithmic space.

For completeness, we will first describe the algorithm presented in  \cite{BGR}. We defer the analysis of this algorithm to the original paper. 

Recall that Theorem C states

\begin{theorem}
Let $S_{\epsilon}$ be a computable dynamical system defined by a continuous function $f$ from a compact space $M$ to itself and a Gaussian noise kernel $p_{f(x)}^{\epsilon}(\cdot)$. Assume also that $f$ is polynomial-time integrable (i.e. it is possible to integrate the convolution of powers of $f$ with polynomial functions in polynomial time). Then computing $\mu$ to within precision $\delta < O(\epsilon)$ requires time and space $O_{S,\epsilon}(\poly(\log(1/\delta)))$. 
\end{theorem}

The algorithm used in the proof of Theorem C proceeds as follows.

\begin{enumerate}
\item
Begin by partitioning $M$ into $A$ regions $\mathfrak{a}_i$ each with diameter at most $\epsilon$. Assign each atom a center $x_i \in \mathfrak{a}_i$. 

\item
Let $\mu^{(t)}(x)$ be the probability density function of the system at time $t$ (given some arbitrary initial distribution $\mu^{(0)}(x)$). Then on each of the regions $\mathfrak{a}_{i}$, $\mu^{(t)}(x)$ can be written as a Taylor series in $(x-x_i)$. In particular, we have that

\begin{equation*}
\mu^{(t)}(x) = \sum_{i=1}^{A}\mathbf{1}\{x \in \mathfrak{a}_i\} \sum_{k=0}^{\infty} \rho_{i,k}^{(t)}(x-x_i)^k
\end{equation*}

\noindent
where $\rho_{i,k} \in \R$ are the coefficients of these Taylor series. The coefficients at time $t+1$ are related to the coefficients at time $t$ via the following linear map.

\begin{equation*}
\rho_{i,l}^{(t+1)} = \sum_{j, m}\rho_{j, m}^{(t)}\int_{\mathfrak{a}_j}(y-x_j)^{m}\frac{\partial^{l}_2 K_{f}(y, x_i)}{l!}dy
\end{equation*}

Call this linear map $P$. The coefficients of $P$ can then be computed to arbitrary precision by computing convolutions of derivatives of the noise kernel with certain polynomials (which is possible in polynomial time by our assumption).

\item
For any positive integer $N$, ignoring all terms in the Taylor expansion of degree larger than $N$ truncates the transition map $P$ to form a finite linear map $P_{N}$ (representable as an $AN$ by $AN$ matrix). The analysis in  \cite{BGR} proves the following lemma. 

\begin{lemma}
There exist log-space computable functions $t(\delta)$ and $N(\delta)$ such that

\begin{equation*}
|| \pi - P^{t(\delta)}_{N(\delta)}\rho||_{\infty} \leq \delta
\end{equation*}

\noindent
for all $\delta > 0$, uniformly in $\rho$, where

\begin{eqnarray*}
t(\delta) &=& O\left(\log \delta^{-1}\exp\left(\epsilon^{-2}\right)\right) \\
N(\delta) &=& O\left(\log \delta^{-1}\,\poly\left(\epsilon^{-1}\right) \right)
\end{eqnarray*}
\end{lemma}
\begin{proof}
See Theorem 36 in  \cite{BGR}. Explicit expressions for $t(\delta)$ and $N(\delta)$ can be found in the proof of Theorem 36.
\end{proof}

By repeated squaring, we can compute $P_{N(\delta)}^{t(\delta)}$ in time $O(\poly(N(\delta))\log t(\delta)) = O_{\epsilon}(\poly(\log \delta^{-1}))$. The above lemma implies that the measure given by $P_{N(\delta)}^{t(\delta)}$ is within $\delta$ of the invariant measure $\mu$, as desired.
\end{enumerate}

We now proceed to prove Theorem \ref{ubnd}. As in Theorem C in \cite{BGR}, we initially restrict ourselves to the one-dimensional case for clarity. We later describe the changes necessary for the $d$-dimensional case.

\newtheorem*{thm:ubnd}{\bf Theorem \ref{ubnd}}
\begin{thm:ubnd} {\em
Let $X=[0,1]$. 
If the noise $p_{f(x)}^{\epsilon}(\cdot)$ is Gaussian, and $f$ is $(\log \frac{1}{\epsilon})+$log-space integrable, then the computation of the invariant measure $\mu$ at precision $\delta$ can be done in space $O\left(\poly\left(\log\frac{1}{\epsilon} + \log\log\frac{1}{\delta}\right)\right)$.}
\end{thm:ubnd}
\begin{proof}
We describe how to adapt the algorithm presented above so that it can be performed in poly-logarithmic space. 

In order to show we can execute the above approach in polylogarithmic space, we must show we can both compute the coefficients of the matrix $P$ to within $\poly(\delta)$ accuracy and that we can then subsequently exponentiate the truncated matrix $P_{N(\delta)}$ to the power $t(\delta)$. Note that the coefficients of $P$ are given by the expression

\begin{equation*}
P^{(i,j)}(l, m) = \int_{\mathfrak{a}_j}(y-x_j)^m\frac{\partial^{l}_2 K_{f}(y, x_i)}{l!}dy
\end{equation*}

\noindent
In the case of a Gaussian kernel, 

\begin{equation*}
K_{f}(y, x_i) = C_{\epsilon}(x_i)\frac{1}{\epsilon\sqrt{2\pi}}\exp\left(-(f(y)-x_i)^2/2\epsilon^2\right)
\end{equation*}

We can expand this expression out via the Taylor series for $\exp(x)$. Since $(f(y) - x_i)$ is bounded (by the diameter of $M$, for example), to approximate this integral to within $\delta$, it suffices to take the first $\poly\left(\frac{1}{\epsilon}+\log\frac{1}{\delta}\right)$ terms of this expansion. We can therefore approximate $P^{(i,j)}(l,m)$ as a linear combination of $\poly\left(\frac{1}{\epsilon}+\log\frac{1}{\delta}\right)$ terms of the form 

\begin{equation*}
\int_{\mathfrak{a}_j}(y-x_j)^mf(y)^k dy
\end{equation*}

By our assumption, we can evaluate each of these integrals (to within precision $\poly(\delta)$) in space $O(\log\log \frac{1}{\delta})$. The coefficients of the linear combination can also each be computed in space complexity $O\left(\poly\left(\log \frac{1}{\epsilon} + \log\log\frac{1}{\delta}\right)\right)$ via the comments in Section \ref{sect:spacebound} (in particular, \ref{facts} and \ref{polys}), and hence the entire linear combination can be computed in this space complexity. The normalization constant $C_{\epsilon}(x_i)$ can similarly be computed in this space complexity by expanding out $\exp(y - x_i)^2$ as a Taylor series in $y$ and integrating over $[0,1]$. 

Finally, we must compute $P_{N(\delta)}^{t(\delta)}$. Note that since $t(\delta)$ is exponential in $\epsilon^{-2}$, we cannot compute $P_{N(\delta)}^{t(\delta)}$ in space $O(\poly(\log \log \delta^{-1} + \log\epsilon^{-1}))$ via repeated squaring, as in the proof of Theorem C. Instead, we apply the algorithm presented in Section \ref{sect:matpow}; by Theorem \ref{matpowers}, this allows us to compute $P_{N(\delta)}^{t(\delta)}$ to within precision $\delta$ in polylogarithmic space. 

Since $P_{N(\delta)}^{t(\delta)}$ is an $AN(\delta)$ by $AN(\delta)$ matrix, where $A = O(\poly(1/\epsilon))$ and $N(\delta) = O(\poly(\log(1/\delta) + 1/\epsilon))$, it follows that computing $P_{N(\delta)}^{t(\delta)}$ can be done in total space complexity 

$$O(\poly(\log AN(\delta))) = O\left(\poly\left(\log \frac{1}{\epsilon} + \log\log\frac{1}{\delta}\right)\right).$$

\end{proof}

\begin{remark}\label{rem:dimension}
To extend this result to the case of $d$ dimensions, we can follow essentially the same procedure; the only change is that we now must write the density functions $\mu^{(t)}(x)$ as multivariate Taylor series in $(\mathbf{x} - \mathbf{x_i})$ (and each of the components of $f$ must be $(\log \epsilon^{-1}) + \log$-space integrable). Since the number of terms in the multivariate Taylor expansion of degree at most $N$ is $O(N^{d})$, the truncated matrix $P_{N}$ still has size polynomial in $\frac{1}{\epsilon}$ and $\log\frac{1}{\delta}$, so the invariant measure can be computed in space complexity

\begin{equation}
\poly\left(d + \log \frac{1}{\epsilon} + \log\log\frac{1}{\delta}\right).
\end{equation}
\end{remark}

\subsection{Computing Taylor coefficients of $f$}

Theorem \ref{ubnd} relies on the assumption that convolutions of powers of $f$ with polynomials are log-space integrable. While this assumption holds true for many natural choices of $f$, it is perhaps not the easiest condition to work with, and one might hope for a more natural constraint on $f$. In this section, we show an alternate constraint which implies our previous assumption; namely, that $f$ is log-space computable, smooth, and has bounded Taylor coefficients. Recall that $f$ is logspace computable
if given $x$ on the input tape, $f(x)$ can be computed within precision $2^{-n}$ using space $O(\log n)$. 
We prove the following theorem\footnote{A similar theorem holds under the assumption that $f$ is computable in polylogarithmic space. The conclusion is then that the integrals are also computable in polylogarithmic space --- which suffices to obtain the conclusion of Theorem~\ref{ubnd}.}. 

\begin{theorem}\label{analytic}
Let $f$ be a function that is log-space computable, smooth, and for some constant $\eta$, satisfies (for all $x$)

\begin{equation*}
|\partial^{k}f(x)| \leq k!\eta^{k}
\end{equation*}

\noindent
Then it is possible to compute integrals of the form

\begin{equation*}
\int_{\mathfrak{a}_j} (y-x_{j})^mf(y)^k dy
\end{equation*}

\noindent
where $\diam \mathfrak{a}_j < \frac{1}{2\eta}$ to within precision $\delta$ in space logarithmic in $m$, $k$, $\log \eta$ and $\log 1/\delta$.
\end{theorem}

\begin{remark}
We note that if $f$ is analytic in $[a,b]$, then such a constant $\eta$ always exists. In fact, if we let $\rho$ to be a strict lower bound of the set of all the radii of convergence of the Taylor series with centers in $[a,b]$ (note that $\rho>0$ by compactness), then the integral Cauchy formula implies that, for any $x\in[a,b]$

\begin{equation*}
|\partial^{k}f(x)| \leq \frac{Mk!}{\rho^{k}}
\end{equation*}
where $M$ is any upper bound of $f$ over $[a,b]^{\rho}=\{z\in\mathbb{C}: |z-x|\leq\rho \text{ for some }x\in [a,b]\}$. 
\end{remark}

To prove the above theorem, we first show that if $f$ satisfies the above constraints, then it is possible to compute its Taylor coefficients in logarithmic space.

\begin{lemma} \label{taylor}
Assume $f$ is log-space computable, smooth, and satisfies $|\partial^{k}f(x)| \leq k!\eta^{k}$ for all $x$ in the domain. Then for any $x_c$, we can write

\begin{equation*}
f(x) = \sum_{k} a_{k}(x-x_c)^k
\end{equation*}

\noindent
The value of $a_{k}$ is then computable to within precision $\delta$ in space logarithmic in $k$, $\log 1/\eta$, and $\log 1/\delta$. 
\end{lemma}
\begin{proof}
We claim that if we choose $\tau = \delta \eta^{-(k+1)}k^{-(k+2)}2^{-k}$, then 

\begin{equation}\label{coeffbound}
\left| \dfrac{\sum_{i=0}^{k}f(x+i\tau)(-1)^{k-i}\binom{k}{i}}{k!\tau^k} - a_k\right| \leq \delta
\end{equation}

Note that since $f$ is computable in log-space, the quantity on the LHS of equation \ref{coeffbound} is computable in space $O(\log k + \log\log 1/\tau) = O(\log k + \log\log \eta + \log\log 1/\delta)$, as desired.

To prove equation \ref{coeffbound}, recall that the Lagrange remainder theorem for Taylor series says that for any $x$ (within the radius of convergence of the Taylor series about $x_c$), we can write

\begin{equation*}
f(x) = \left(\sum_{i=0}^{k} a_{i}(x-x_c)^i\right) + \dfrac{f^{(k+1)}(\xi)}{(k+1)!}(x-x_c)^{k+1}
\end{equation*} 

\noindent
for some $\xi$ between $x_c$ and $x$. Write $f(x) = \left(\sum_{i=0}^{k} a_{i}(x-x_c)^i\right) + R_{k+1}(x)$. By our constraint on $f$, we know that $\left|\frac{f^{(k+1)}(\xi)}{(k+1)!}\right| \leq \eta^{k+1}$, so we can rewrite this as

\begin{equation*}
\left| R_{k+1}(x) \right| \leq \eta^{k+1}(x-x_c)^{k+1}
\end{equation*} 

Next, recall the following binomial identities. For all $r < k$, we have that

\begin{equation*}
\sum_{i=0}^{k}i^{r}(-1)^{k-i}\binom{k}{i} = 0
\end{equation*}

\noindent
On the other hand, when $r=k$, we have that

\begin{equation*}
\sum_{i=0}^{k}i^{k}(-1)^{k-i}\binom{k}{i} = k!
\end{equation*}

Substituting in the Taylor expansion for $f$ and applying the above binomial identities, we see that

\begin{eqnarray*}
\left|\dfrac{\sum_{i=0}^{k}f(x+i\tau)(-1)^{k-i}\binom{k}{i}}{k!\tau^k} - a_k\right| &=& \left|\dfrac{\sum_{i=0}^{k}R_{k+1}(x+i\tau)(-1)^{k-i}\binom{k}{i}}{k!\tau^k} \right| \\
&\leq & \dfrac{1}{k!\tau^k}\sum_{i=0}^{k}\left|\eta^{k+1}i^{k+1}\tau^{k+1}\binom{k}{i}\right| \\
&\leq & \dfrac{\tau \eta^{k+1}}{k!}\sum_{i=0}^{k}k^{k+1}2^{k} \\
&\leq & \dfrac{\tau \eta^{k+1}k^{k+2}2^{k}}{k!} \\
&\leq & \delta
\end{eqnarray*}

\noindent
as desired.
\end{proof}

We can now prove Theorem \ref{analytic}.

\begin{proof}[Proof of Theorem~\ref{analytic}]
We wish to compute the integral

\begin{equation*}
\int_{\mathfrak{a}_j}(y-x_j)^{m}f(y)^k dy
\end{equation*}

\noindent where we know that $\diam \mathfrak{a}_j < \frac{1}{2\eta}$. Write $f(y) = \sum a_i(y-x_j)^i$; by our assumption $|a_i| \leq \eta^{i}$ for all $i$. 

Let $f_{M}(y) = \sum_{i=0}^{M} a_i(y-x_j)^i$. Then we have that 

\begin{eqnarray*}
|f(y) - f_M(y)| &=& \left |\sum_{i=M+1}^{\infty} a_i(y-x_j)^{i}\right|\\
&\leq & \sum_{i=M+1}^{\infty}|a_i|\cdot |y-x_j|^{i} \\
&\leq & \sum_{i=M+1} \eta^{i}(2\eta)^{-i} \\
&=& \sum_{i=M+1} 2^{-i} \\
&=& 2^{-M}
\end{eqnarray*}

Since $f(y) \in [0,1]$, this further implies that $|f(y)^k - f_{M}(y)^k| \leq k2^{-M}$; it follows that if we take $M = \log(\delta/k)$, then $|f(y)^{k} -f_{M}(y)^{k}| \leq \delta$, and in particular

\begin{equation*}
\left| \int_{\mathfrak{a}_j}(y-x_j)^{m}f(y)^k dy - \int_{\mathfrak{a}_j}(y-x_j)^{m}f_M(y)^k dy\right| \leq \delta
\end{equation*}

But note that by Lemma \ref{taylor}, we can compute each of the coefficients of $f_{M}(y)$ (to within precision $\poly(\delta)$) in space logarithmic in $M$, $\log 1/\eta$, and $\log 1/\delta$. We can then compute the coefficients of $(y-x_j)^mf_{M}(y)^k$ via remark \ref{polys} of Section \ref{sect:spacebound}, and hence compute the integral over $\mathfrak{a_j}$ in space $O(\log k + \log m + \log \log \eta + \log\log 1/\delta)$, as desired.
\end{proof}

Theorem \ref{ubndalt} now follows as a straightforward corollary to Theorem \ref{analytic}.

\newtheorem*{thm:ubndalt}{\bf Theorem \ref{ubndalt}}
\begin{thm:ubndalt} {\em
Let $X=[0,1]$. 
If the noise $p_{f(x)}^{\epsilon}(\cdot)$ is Gaussian, and $f$ is log-space computable, smooth, and (for some $\eta>0$) satisfies $|\partial^{k}f(x)| \leq k!\eta^{k}$ for all $x$, then the computation of the invariant measure $\mu$ at precision $\delta$ can be done in space $O\left(\poly\left(\log \eta + \log\frac{1}{\epsilon} + \log\log\frac{1}{\delta}\right)\right)$.  }
\end{thm:ubndalt}
\begin{proof}
In the proof of Theorem \ref{ubnd}, we make the slight modification that instead of simply picking the regions $\mathfrak{a}_i$ to satisfy $\diam \mathfrak{a}_i \leq \epsilon$, we instead make them satisfy the stronger requirement that $\diam \mathfrak{a}_i \leq \min(\epsilon, 1/(2\eta))$. Then, by Theorem \ref{analytic}, we can compute all the necessary integrals in logarithmic space, as before by the earlier assumption. 

Since the number of regions is linear in $\eta$, the resulting space bound is $O(\poly(\log\eta + \log \frac{1}{\epsilon} + \log\log\frac{1}{\delta}))$, as desired.
\end{proof}

\section{Space lower bound for computing invariant measures}\label{sect:lbnd}

In this section we prove Theorem \ref{lbnd}. We begin by proving a weaker version of Theorem \ref{lbnd} where we don't restrict our constructed function $f$ to be analytic (or even continuous). 

\begin{lemma}\label{lbndpre}
Any algorithm that can compute the invariant measure $\mu$ of a dynamical system to within precision $\delta$ with Gaussian noise kernel $p_{f(x)}^{\epsilon}(\cdot)$ requires space at least $\Omega\left(\log\frac{1}{\epsilon} + \log\log\frac{1}{\delta}\right)$.
\end{lemma}

\begin{proof}
Since our output is of size $\log\frac{1}{\delta}$, it requires $\Omega\left(\log\log\frac{1}{\delta}\right)$ space to simply keep track of which bit we are currently outputting. This immediately shows the $\Omega\left(\log\log\frac{1}{\delta}\right)$ part of the lower bound.

It remains to show the $\Omega\left(\log\frac{1}{\epsilon}\right)$ portion of the lower bound. We will present a $SPACE(\log M)$-reduction from $SPACE(M)$ to the problem of computing the invariant measure of a noisy dynamical system $S_{\epsilon}$ with $\epsilon = 2^{-\Theta(M)}$, thus showing computing the invariant measure of a noisy dynamical system requires space at least $\Omega(\log\frac{1}{\epsilon})$.

More specifically, we will show how to convert any Turing machine $T$ with tape size $M$ along with an input $s$ into a function $f:X\rightarrow X$ that `embeds' this machine/input pair. We will construct this embedding so that the invariant measure of the corresponding dynamical system will have significant measure on some subset of the domain $X$ if $T$ accepts $s$ and close to zero measure otherwise. 

Let $S$ be the total number of states of the Turing machine $T$ (including the current state of the tape, so $S = \Theta(2^{M})$), and let $N = 2S^2$. Choose $X$ to be the unit interval $[0, 1]$, and partition $X$ into the $N$ intervals $X_{k} = [\frac{k}{N}, \frac{k+1}{N}]$ for $0 \leq k < N$. Let $c_k = \frac{2k+1}{2N}$ be the center of interval $X_k$. 

Choose $\epsilon$ (the size of the Gaussian noise) so that $\int_{-1/2N}^{1/2N}p_{\epsilon}(x)dx = 1-N^{-100}$; since the tail of a Gaussian decreases to $0$ exponentially quickly, it suffices to take $\epsilon = \Omega(N^{-2}) = 2^{-O(M)}$ (then this integral corresponds to the probability of being at least $\Omega(N)$ standard deviations away from the mean).

Finally, if $x\in X_{k}$, then we define $f$ so that $f(x) = c_{\suc(k)}$, where $\suc(k):\{0, \dots, N-1\} \rightarrow \{0, \dots, N-1\}$ is defined as follows.

\begin{enumerate}[(i)]
\item
If $k < S^2$, set $(v, t) = \left(\lfloor\frac{k}{S}\rfloor, k - S\lfloor\frac{k}{S}\rfloor\right)$. We will interpret $v$ as the binary representation of some state of $T$, and $t$ as a counter of how many steps we have run machine $T$ for so far. 

\begin{enumerate}
\item
If $v$ is an accepting state, set $\suc(k) = S^2$. 
\item
If $v$ is a rejecting state, set $\suc(k) = sS$, where $s$ is the initial state of the Turing machine $T$.
\item
If $t<S-1$ and $v$ is neither an accepting or a rejecting state, find the successor state $v'$ of $v$ according to the Turing machine $T$ (note that since computation is local, this can be done in space $O(\log M)$), and set $\suc(k) = v'S + (t+1)$. 
\item
If $t=S-1$, set $\suc(k) = sS$, where $s$ is the initial state of the Turing machine $T$.
\end{enumerate}

\item
If $S^2 \leq k < 2S^2 - 1$, then $\suc(k) = k+1$. 
\item
If $k = 2S^2-1$, then $\suc(k) = sS$, where $s$ is the initial state of $T$. 
\end{enumerate}

Intuitively, this function $f$ simulates the Turing machine $T$ for up to $S$ time steps (the maximum amount of time a Turing machine with $S$ states can take to reach an accepting state). If, within these $S$ time steps, we encounter an accepting state, we go on a walk for another $S$ time steps through $[1/2, 1]$ and then return to the initial state; otherwise, if we encounter a rejecting state (or run for $S$ steps without accepting or rejecting), we immediately return to the initial state. In this way, if $s$ is an accepting initial state, the invariant measure will have approximately half their weight on the interval $[1/2, 1]$, and if $s$ is not an accepting initial state, the invariant measure will have approximately no weight on $[1/2, 1]$. We formalize this intuition below.

Let $\mu$ be the invariant measure of this dynamical system perturbed by Gaussian noise of variance $\epsilon^2$ with $\epsilon$ as chosen above (note that since the noise is Gaussian, there must be a unique invariant measure; this follows from the fact that for any set $U$ of positive measure, the probability $x_{t+1} \in U$ given $x_{t}$ is always strictly positive). We claim that if $T$ eventually accepts on $s$, then $\mu$ will have measure at least $1/3$ on $[1/2, 1]$. Otherwise, $\mu$ will have measure approximately $0$ on $[1/2, 1]$.

Let $\mathcal{S} = \{sS, \suc(sS), \suc(\suc(sS)), \dots \}$ be the set of iterates of the initial state $s$ of our Turing machine under this successor function. Note that if $T$ accepts starting on $s$, then $\{S^2, \dots, 2S^2-1\}$ is a subset of $\mathcal{S}$; otherwise, if it rejects or fails to halt, then $\{S^2, \dots, 2S^2-1\}$ is not a subset of $\mathcal{S}$.

We first claim that the weight of the invariant measure $\mu$ over states in $\mathcal{S}$ is at least $1-N^{-99}$. To see this, let $x_1, x_2, \dots$ be a sequence of iterates of our dynamical system. Call a time $t$ \textit{bad} if $x_{t} \in X_{k}$ but $x_{t+1} \not\in X_{\suc(k)}$. By our choice of $\epsilon$, the probability of any given time $t$ being bad is at most $N^{-100}$ and is independent of all other times being bad. In addition, by our construction, after $N$ noise-free steps we are guaranteed to be in $\mathcal{S}$, since after $N$ steps of $\suc(k)$ we must pass through $sS$. If we let $X_{\mathcal{S}} = \cup_{k\in \mathcal{S}} X_k$, it then follows that the probability that $x_{t} \in X_{\mathcal{S}}$ is at least $(1-N^{-100})^{N} \geq 1-N^{-99}$. 

Next, assume that $T$ accepts on $s$, and let $X_{path} = \cup_{k=S^2}^{2S^2-1}X_{k} = [1/2, 1]$; note that $X_{path}$ is a subset of $X_{\mathcal{S}}$. We claim that the weight under the measure $\mu$ of $X_{path}$ is at least $\frac{1}{2}(1-2S^{-9})$ of the weight of $X_{\mathcal{S}}$. To see this, call the sequence $x_{t}, x_{t+1}, \dots, x_{t+|\mathcal{S}|}$ \textit{good} if no time $t+i$ is bad for any $0\leq i < |\mathcal{S}|$ (in other words, no low probability noise events occur for $|\mathcal{S}|$ steps). Note that this occurs with probability at least $(1-N^{-100})^{N} \geq 1-N^{-99}$. But in any good sequence, each element of $\mathcal{S}$ appears exactly once; it follows that, asymptotically, the probability that $x_{t}$ belongs to $X_{path}$ given that $x_{t}$ belongs to $X_{\mathcal{S}}$ is at least

\begin{equation*}
(1-N^{-99})\frac{S^2}{|\mathcal{S}|} \geq (1-N^{-99})\frac{S^2}{2S^2} = \frac{1}{2}(1-N^{-99})
\end{equation*}

Combining these two results, it follows that the weight of the invariant measure over $X_{path}$ is at least

\begin{equation*}
\frac{1}{2}(1-S^{-99})^2 > \frac{1}{3}
\end{equation*}

\noindent
On the other hand, if $T$ does not accept on $s$, then $[1/2, 1] \cap X_{\mathcal{S}} = \emptyset$, and therefore the weight of $\mu$ over $[1/2, 1]$ is at most $N^{-99} \ll 1/3$, as desired.

\end{proof}

Note that, since the function constructed in this reduction is piecewise linear with $O(2^M)$ pieces, it is in fact $(\log \epsilon^{-1}) + \log$-space integrable in the sense of Theorem \ref{ubnd}. On the other hand, this function is not continuous (let alone analytic), and hence does not satisfy the conditions of Theorem \ref{ubndalt}.  

To prove Theorem \ref{lbnd}, we transform the above example into a uniformly analytic function by replacing each of the intervals in the construction in Lemma \ref{lbndpre} with an analytic approximation to a step function. We describe this below, starting with the construction of our analytic `step function'.

\begin{lemma}\label{step}
For any $\alpha, \beta > 0$, there exists an analytic function $F(x):\R\rightarrow\R$ that satisfies the following constraints:

\begin{itemize}
\item For all $x < -\alpha$, $|F(x)| < \beta$.
\item For all $x > \alpha$, $|F(x)-1| < \beta$.
\item For all integer $k \geq 0$ and all $x$, $|\partial^{k}F(x)| \leq k!\eta^{k}$ for some $\eta = O(\alpha^{-1}\log\beta^{-1})$.
\item The function $F(x)$ is computable to within precision $\delta$ in space $O(\log\log \delta^{-1})$. 
\end{itemize}
\end{lemma}
\begin{proof}
We will consider functions of the form

\begin{equation}
F(x) = \frac{1}{1+e^{-Cx}}
\end{equation}

\noindent
where $C$ is a positive integer. Note that in order for $|F(x)-1|$ to be less than $\beta$ for all $x>\alpha$, we must have

\begin{equation*}
\left|\frac{1}{1+e^{-C\alpha}} -1\right| < \beta
\end{equation*}

\noindent
which is satisfied when
\begin{equation*}
C > \alpha^{-1}\log\frac{1-\beta}{\beta}
\end{equation*}

Likewise, in order for $|F(x)|$ to be less than $\beta$ when $x<-\alpha$, we must have

\begin{equation*}
\left|\frac{1}{1 + e^{C\alpha}}\right| < \beta
\end{equation*}

\noindent
which is satisfied when

\begin{equation*}
C > \alpha^{-1}\log\frac{1-\beta}{\beta}
\end{equation*}

\noindent
Therefore to satisfy the first two requirements, we can take

\begin{equation*}
C = \left\lceil \alpha^{-1}\log\frac{1-\beta}{\beta} \right\rceil \approx \alpha^{-1}\log\beta^{-1}
\end{equation*}

To prove the third requirement, note that we can write

\begin{equation*}
F(x) = \frac{1}{2}\left(1 + \tanh\left(\frac{Cx}{2}\right)\right)
\end{equation*}

\noindent
By \cite{AS65}, it is known that (for $x\geq 0$),

\begin{eqnarray*}
\left|\dfrac{d^k\tanh(x)}{dx^k}\right| &=& \frac{2^{k+1}e^{2x}}{(1+e^{2x})^{k+1}}\left|\sum_{j=0}^{k-1}\left\langle {k \atop j} \right\rangle (-1)^{j} e^{2jx}\right| \\
&\leq & \frac{2^{k+1}e^{2(k+1)x}}{(1+e^{2x})^{k+1}}\sum_{j=0}^{k-1}\left\langle {k \atop j} \right\rangle \\
&=& 2^{k+1}\left(\frac{e^{2x}}{1+e^{2x}}\right)^{k+1} k! \\
&\leq & 2^{k+1} k!
\end{eqnarray*}

\noindent
where $\left\langle {n \atop i} \right\rangle$ are Eulerian numbers of the second kind (in the third line we use the fact that $\sum_{i} \left\langle {n \atop i} \right\rangle = n!$). Since $\tanh(x)$ is an odd function, the same bound holds for $x \leq 0$. It follows that for all $k > 0$ and all $x$,

\begin{equation*}
|\partial^{k}F(x)| \leq C^{k} k!
\end{equation*}

\noindent
and therefore we can take $\eta = C$ (for $k=0$, it suffices to note that $|F(x)| \leq 1$ for all $x$). 

Finally, since we can compute $e^{x}$ to within precision $\delta$ in space $O(\log\log \delta^{-1})$ via Lemma \ref{compexp}, and since we can perform all arithmetic operations to within precision $\delta$ in space $O(\log\log\delta^{-1})$ via the remarks in Section \ref{sect:spacebound}, it is possible to compute $F(x)$ in space $O(\log\log\delta^{-1})$.
\end{proof}

We now proceed to prove Theorem \ref{lbnd}.

\newtheorem*{thm:lbnd}{\bf Theorem \ref{lbnd}}
\begin{thm:lbnd}{\em
Any algorithm that can compute the invariant measure $\mu$ to within precision $\delta$ of a dynamical system with Gaussian noise kernel $p_{f(x)}^{\epsilon}(\cdot)$ and analytic transition function $f(x)$ (that uniformly satisfies $|\partial^{k}f(x)| \leq k!\eta^k$ for some $\eta = \poly(\epsilon^{-1})$) requires space at least $\Omega\left(\log\frac{1}{\epsilon} + \log\log\frac{1}{\delta}\right)$.}
\end{thm:lbnd}
\begin{proof}
We will use the function $F(x)$ defined in Lemma \ref{step} to approximate the function $f(x)$ defined in the proof of Lemma \ref{lbndpre} with an analytic function. We will then show that the dynamical system corresponding to this new $f$ still has the property that it has significant measure on the interval $[1/2, 1]$ if and only if the Turing machine $T$ accepts $s$. 

As before, let $S = 2^{M}$ be the number of states of the Turing machine $T$, and let $N = 2S^2$. Partition the interval $[0,1]$ into the $N$ intervals $X_{k} = [\frac{k}{N}, \frac{k+1}{N}]$ for $0 \leq k < N$, and let $c_k = \frac{2k+1}{2N}$ be the center of interval $X_k$. Let $\suc(k)$ be defined equivalently as in the proof of Theorem \ref{lbndpre}. Then, in Lemma \ref{step}, set $\alpha = \beta = S^{-100}$, and consider the function

\begin{equation}
f(x) = c_{\suc(0)} + \sum_{i=1}^{N} \left(c_{\suc(i)} - c_{\suc(i-1)}\right)F\left(x-\frac{i}{N}\right)
\end{equation}

Note that by Lemma \ref{step}, this function $f$ satisfies the following condition: if $|x - c_k| \leq \frac{1}{2N} - \alpha$, then $|f(x) - c_{\suc(k)}| \leq N\beta = O(S^{-98})$. We will next claim that if we set $\epsilon = S^{-10}$, then we simultaneously have that

\begin{equation}\label{eq:pbdtogd}
\max_{x} p_{\epsilon}(x) \leq \frac{S^{10}}{\sqrt{2\pi}}
\end{equation}

\noindent
and that

\begin{equation}\label{eq:pgdtobd}
\int_{-(\frac{1}{2N} - \alpha - N\beta)}^{\frac{1}{2N} - \alpha - N\beta}p_{\epsilon}(x) dx \geq 1-16 S^{-16}
\end{equation}

To show the first of these inequalities, note simply that $p_{\epsilon}(x) \leq \frac{1}{\epsilon\sqrt{2\pi}}$; inequality \ref{eq:pbdtogd} then follows from substituting $\epsilon = S^{-10}$. To show the second inequality, note first that $\frac{1}{2N} - \alpha - N\beta \geq \frac{1}{4N}$. Hence the integral in inequality \ref{eq:pgdtobd} is at most the probability that the noise is within $\frac{1}{4N\epsilon} = \frac{S^{8}}{4}$ standard deviations of its mean. By Chebyshev's inequality it follows that this probability is at most $1 - 16S^{-16}$, from which this second inequality follows (much better bounds are in fact possible).

We can now proceed to analyze the invariant measure $\mu$ of this dynamical system. For each $k$, let $Y_{k} = \left[c_k - \frac{1}{2N} + \alpha, c_k + \frac{1}{2N} - \alpha\right]$, and let $Y = \cup_{k=0}^{N-1}Y_{k}$. We will first show that $\mu$ has measure at least $1-32S^{-16}$ on $Y$. 

Let $x_1, x_2, \dots$ be a sequence of iterates of this dynamical system. Call a time $t$ \textit{bad} if $x_{t} \in Y_{k}$ but $x_{t+1} \not\in Y_{\suc(k)}$. By inequality \ref{eq:pbdtogd}, the probability that a time $t$ is bad (given that $x_{t} \in Y_k$ for some $k$) is at most $16S^{-16}$. It follows that $\mathrm{Pr}[x_{t+1} \not\in Y| x_{t}\in Y] \leq 16S^{-16}$. On the other hand, note that if $x_{t} \not\in Y$, then by inequality \ref{eq:pgdtobd}, the probability $x_{t+1}$ is in $Y$ is at least

\begin{eqnarray*}
1 - \left(\max_{x} p_{\epsilon}(x)\right) |X \setminus Y| &\geq & 1 - \frac{S^{10}}{\sqrt{2\pi}}(N\alpha) \\ 
& \geq & 1 - \sqrt{\frac{2}{\pi}}S^{-88} \\
& \geq & \frac{1}{2}
\end{eqnarray*}

\noindent
It follows that the weight of $\mu$ over $Y$ must be at least $0.5/(0.5+16S^{-16}) \geq 1-32S^{-16}$, as desired.

Next, as before, let $\mathcal{S} = \{sS, \suc(sS), \suc(\suc(sS)), \dots \}$ be the set of iterates of the initial state $s$ of our Turing machine. If $T$ accepts starting on $s$, then $\{S^2, \dots, 2S^2-1\}$ is a subset of $\mathcal{S}$; otherwise, if it rejects or fails to halt, then $\{S^2, \dots, 2S^2-1\}$ is not a subset of $\mathcal{S}$. Let $Y_{\mathcal{S}} = \cup_{k\in \mathcal{S}} Y_k$. We will next show that the weight of $\mu$ over $Y_{\mathcal{S}}$ is at least $1 - 64S^{-14}$.

To prove this, recall that if we start at some $x \in Y$, after $N$ noise-free steps, we are guaranteed to be in $Y_{\mathcal{S}}$. Since the weight of $\mu$ over $Y$ is at least $1-32S^{-16}$ and since the probability a string of $N$ steps are all good is at least $(1-16S^{-16})^{N} \geq 1 - 32S^{-14}$, the weight of $\mu$ over $Y_{\mathcal{S}}$ is at least $(1-32S^{-16})(1-32S^{-14}) \geq 1-64S^{-14}$. 

Finally, assume that $T$ accepts on $s$, and let $Y_{path} = \cup_{k=S^2}^{2S^2-1}Y_{k} = [1/2, 1]$; note that $Y_{path}$ is a subset of $Y_{\mathcal{S}}$. We claim that the weight under the measure $\mu$ of $Y_{path}$ is at least $\frac{1}{2}(1-32S^{-14})$ of the weight of $Y_{\mathcal{S}}$. To see this, call the sequence $x_{t}, x_{t+1}, \dots, x_{t+|\mathcal{S}|}$ \textit{good} if no time $t+i$ is bad for any $0\leq i < |\mathcal{S}|$. Note that this occurs with probability at least $(1-16S^{-16})^{N} \geq 1-32S^{-14}$. But in any good sequence, each element of $\mathcal{S}$ appears exactly once; it follows that, asymptotically, the probability that $x_{t}$ belongs to $X_{path}$ given that $x_{t}$ belongs to $X_{\mathcal{S}}$ is at least

\begin{equation*}
(1-32S^{-14})\frac{S^2}{|\mathcal{S}|} \geq (1-32S^{-14})\frac{S^2}{2S^2} = \frac{1}{2}(1-32S^{-14})
\end{equation*}

Combining these two results, it follows that the weight of the invariant measure over $Y_{path}$ (and hence $[1/2, 1]$) is at least

\begin{equation*}
\frac{1}{2}(1-32S^{-14})(1-64S^{-14}) > \frac{1}{3}
\end{equation*}

\noindent
On the other hand, if $T$ does not accept on $s$, then $[1/2, 1] \cap Y_{\mathcal{S}} = \emptyset$, and therefore the weight of $\mu$ over $[1/2, 1]$ is at most $4S^{-14} \ll 1/3$, as desired.

\end{proof}








\bibliographystyle{alpha}
\bibliography{bibliography}

\begin{appendix}

\section{Space-Bounded Church-Turing Thesis} \label{sec:sbct}

This appendix serves as a short introduction to the Space-Bounded Church-Turing thesis (hereafter referred to as the SBCT). For more further details, we suggest the reader consult \cite{BRS}.

Let $\mathcal{S} = X_{t}$ be a closed, discrete-time stochastic system over a state space $\mathcal{X}$. We define the \textit{memory available to $\mathcal{S}$} as

\begin{equation}
\mathcal{M}(\mathcal{S}) = \sup_{t} \sup_{\mu} I_{X_{t} \sim \mu}(X_{t}; X_{t+1})
\end{equation}

\noindent
where the inner supremum is taken over all distributions $\mu$ over $\mathcal{X}$. Here, $I(X_{t};X_{t+1})$ is Shannon's mutual information and is a measure of how much information is preserved from time $t$ to time $t+1$. If $f(x, y)$ is the PDF of the distribution of $(X_{t}, X_{t+1})$, then $I(X_{t};X_{t+1})$ is defined via

\begin{equation*}
I(X_{t};X_{t+1}) = \int\int f(x,y)\log\frac{f(x,y)}{f(x)f(y)}dxdy
\end{equation*}

We can now state a concrete form of the SBCT (in \cite{BRS}, the statement below is referred to as the Simulation Assertion). 

\begin{conjecture}\label{SBCT}
The problem of computing the asymptotic behavior of a stochastic system $\mathcal{S}$ with memory $M = \mathcal{M}(\mathcal{S})$ to within precision $2^{-n}$ is in the complexity class $\mathbf{SPACE}((M+\log n)^{O(1)})$. 
\end{conjecture}

While this conjecture can be easily falsified by an artificial construction, a case can be made that it holds for physically relevant systems. 
In particular, the present paper establishes Conjecture \ref{SBCT} in the case where $\mathcal{S}$ is a dynamical system with $\epsilon$-Gaussian noise. 

\begin{lemma}
If $\mathcal{S}$ is a dynamical system over $X = [0,1]$ with Gaussian noise kernel $p_{f(x)}^{\epsilon}$, then $\mathcal{M}(\mathcal{S}) = \Theta(\log\epsilon^{-1})$.
\end{lemma}
\begin{proof}
We can write $I(X_{t};X_{t+1})$ as $H(X_{t+1}) - H(X_{t+1}|X_{t})$ (here $H(x)$ is the differential Shannon entropy; see \cite{CT91}). Let $p(x)$ be the PDF of $X_{t+1}$; by Jensen's inequality, note that

\begin{equation*}
H(X_{t+1}) = -\int_{0}^{1}p(x)\log p(x) dx \leq -\int_{0}^{1} 1\log 1 dx = 0
\end{equation*}

On the other hand, $X_{t+1}|X_{t}$ is a Gaussian with variance $\epsilon$; the differential entropy of such a distribution is given by $\ln(\epsilon\sqrt{2\pi e})$ (\cite{CT91}). Hence it follows that $I(X_{t};X_{t+1}) \leq -\ln(\epsilon\sqrt{2\pi e}) = O(\log \frac{1}{\epsilon})$. 
\end{proof}

\begin{corollary}
Conjecture \ref{SBCT} is true for the case where $\mathcal{S}$ is a dynamical system over $X = [0,1]$ with Gaussian noise kernel $p_{f(x)}^{\epsilon}$.
\end{corollary}

\section{Exponentiating numbers to high powers}\label{sect:numpow}

In 1984, Alt showed how to compute $x^{E}$ (for potentially exponentially large $E$) in logarithmic space by computing $\exp(E\log x)$ (see \cite{Alt84}). For completeness, we include in this appendix a proof of Alt's result. Formally, we prove the following theorem:

\begin{theorem}\label{powers}
Given a positive real number $x$ presented in binary up to precision $2^{-\poly(n)}$ and an integer exponent $E \leq 2^{\poly(n)}$, there exists an algorithm that computes $y = x^{E}$ in space $O(\log n)$ to within precision $2^{-\poly(n)}$ if $y \leq 2^{n}$ (and otherwise reports that $y \geq 2^{n}$).
\end{theorem}

\begin{remark}
The condition that we only output $y$ to within precision $2^{-\poly(n)}$ if $y \leq 2^n$ is crucial. In general, it is possible for $x^{E}$ to be on the order of $2^{2^{\poly(n)}}$, and hence require an exponential number of bits in $n$ to represent to within precision $2^{-\poly(n)}$. However, a machine with space $O(\log n)$ only has $\poly(n)$ different states, and hence cannot hope to output a binary string of length exponential in $n$. 
\end{remark}

Instead of proceeding via iterative squaring, our algorithm first calculates $\log x$ to sufficient precision, then computes $\exp(E\log x)$.  To do this, we first show that we can compute the functions $\log(x)$ and $\exp(x)$ in logarithmic space.

\begin{lemma}\label{complog}
Given a positive real number $x$ up to precision $2^{-\poly(n)}$, we can compute $\log x$ to within precision $2^{-\poly(n)}$ in logarithmic space.
\end{lemma}
\begin{proof}
First, find an integer $v$ such that $2^{-v/2}x \in [1, 1.5)$; if $\ell$ is the length of the binary representation of the integer part of $x$, then we can choose $v$ to be either $2\ell$ or $2\ell+1$. Set $w = 2^{-v/2}x$. Then $\log x = \log w + v\log \sqrt{2}$. 

We have therefore reduced the problem to computing the logarithm of numbers within the range $[1, 1.5]$. To do this, we will use the following Taylor expansion for $\log(1+z)$:

\begin{equation*}
\log(1+z) = \sum_{k=1}^{\infty} (-1)^{k+1}\frac{z^k}{k}
\end{equation*}

Since $z = x-1 \leq  \frac{1}{2}$ for all $x$ in $[1, 1.5]$, the error from truncating after $k$ terms is at most $2^{-k}$. Therefore, to evaluate $\log(1+z)$ to within precision $2^{-\poly(n)}$, it suffices to sum the first $\poly(n)$ terms of this series. By comments \ref{itmult} and \ref{realarith} of section \ref{sect:spacebound}, each of these terms can be evaluated in logarithmic space, and hence $\log x$ can be computed to within precision $2^{-n}$ in space $O(\log n)$, as desired.
\end{proof}

Again, when computing $\exp(x)$ in logarithmic space, we must ensure the output does not have length exponential in $n$. For that reason, we restrict ourselves to computing $\exp(x)$ for values $x < \poly(n)$. 

\begin{lemma}\label{compexp}
Given a positive real number $x < \poly(n)$ to within precision $2^{-\poly(n)}$, we can compute $\exp(x)$ to within precision $2^{-\poly(n)}$ in logarithmic space. 
\end{lemma}
\begin{proof}
Again, we will compute $\exp(x)$ via its Taylor expansion, truncating after a suitable number of terms. Recall that the Taylor expansion of $\exp(x)$ is given by

\begin{equation*}
\exp(x) = \sum_{k=0}^{\infty} \dfrac{x^k}{k!}
\end{equation*}

By Lagrange's remainder theorem for Taylor series, we have that the error from truncating after $k$ terms is at most

\begin{equation*}
\frac{\exp(x)x^{k+1}}{(k+1)!}
\end{equation*}

Choosing $k$ on the order of $x^{d}$ (for some constant $d$) and applying Stirling's approximation, guarantees that this error is at most $2^{-n^{d}}$; it therefore follows that to evaluate $\exp(x)$ to within precision $2^{-\poly(n)}$, it suffices to compute and sum a polynomial number of terms of this series, which can be done in logarithmic space.
\end{proof}

We now proceed to prove Theorem \ref{powers}.

\begin{proof}[Proof of Theorem~\ref{powers}]
As mentioned earlier, we will compute $x^{E}$ by computing $\exp(E\log x)$. Since $e^{x+\delta} - e^{x} \approx \delta e^{x}$, computing $\exp(E\log x)$ to within precision $\epsilon$ requires us to compute $E\log x$ to within precision $\delta = \epsilon/\exp(E\log x)$. Since we need only do this in the case where $\exp(E\log x) \leq 2^{n}$, we can assume $\delta \leq \epsilon 2^{-n} = 2^{-\poly(n)}$. 

Hence, first compute (via Lemma \ref{complog}) $E\log x$ to within this required precision, and then compute (via Lemma \ref{compexp}) $\exp(E\log x)$ to within the desired precision. (If $E\log x \geq n\log 2$, then we can instead output that $y \geq 2^{n}$). This completes the proof.
\end{proof}

Finally, we demonstrate how to extend this result to the case of computing powers of complex numbers. Our approach will be similar, except we will consider the phase and amplitude separately.

We first show it is possible to compute (in logarithmic space) the argument of a complex number.

\begin{lemma}\label{comparg}
Given a complex number $z = x + yi$ where $x$ and $y$ are given to precision $2^{-\poly(n)}$, we can compute $\arg(z)$ to within precision $2^{-\poly(n)}$ in logarithmic space (where $\arg(z)$ is the unique $\theta \in [0, 2\pi)$ such that $z = Me^{i\theta}$ for some positive real $M$).
\end{lemma}
\begin{proof}
Set $\omega = \exp(- 2\pi i/6)$. Find an integer $v \in [0, 12)$ such that $z' = \omega^{v}z$ has an argument in the range $[0, \arctan(0.6))$ (note that $\pi/3 < \arctan(0.6)$). Then, $\arg(z) = \arg(z') - \frac{v\pi}{3}$. 

We therefore wish to compute $\arg(z)$ for those $z$ whose argument lies in $[0, 0.6)$. For these $z$, $\arg(z) = \arctan\left(\frac{y}{x}\right)$, so it therefore suffices to compute $\arctan(t)$ for $t \in [0, 0.6)$. 

As before, we will proceed via evaluating the Taylor series of $\arctan(t)$. Recall that the Taylor expansion of $\arctan(t)$ is given by

\begin{equation*}
\arctan(t) = \sum_{k=0}^{\infty}(-1)^{k}\frac{t^{2k+1}}{2k+1}
\end{equation*}

By comparison to a geometric series, the error from truncating after $k$ terms of this series is at most $2^{-k}$. Therefore, to evaluate $\arg(z)$ to within precision $2^{-\poly(n)}$, it suffices to sum the first $k = \poly(n)$ terms of this series. Since computing each term can also be done in logarithmic space, it is therefore possible to compute $\arg(z)$ to within this precision in logarithmic space.
\end{proof}

\begin{theorem}\label{cpxpowers}
Given a complex number $z = x + iy$ with $x$ and $y$ presented in binary up to precision $2^{-\poly(n)}$ and an integer exponent $E = O(2^{\poly(n)})$, there exists an algorithm that computes $z' = z^{E}$ in space $O(\log n)$ to within precision $2^{-\poly(n)}$ if $|z'| \leq 2^{n}$ (and otherwise reports that $|z'| \geq 2^{n}$).
\end{theorem}
\begin{proof}
Write $z = re^{i\theta}$, with $r = \sqrt{x^2 + y^2}$ and $\theta = \arg(z)$. Note then that $z' = r^{E}e^{iE\theta}$. 

Given $x$ and $y$, we can compute $r$ to any desired precision $2^{-\poly(n)}$ and hence by Theorem \ref{powers}, we can also compute $r^{E}$ to within any precision $2^{-\poly(n)}$ (as long as $r^{E} \leq 2^{n}$). Likewise, by Lemma \ref{comparg}, we can compute $\theta$ to within any precision $2^{-\poly(n)}$, and hence compute $\theta' = E\theta$. By reducing $\theta'$ modulo $2\pi$ so that $\theta' \in [0, 2\pi)$, we can then compute $\exp(i\theta') = \cos(\theta') + i\sin(\theta')$  to within precision $2^{-\poly(n)}$ via similar logic to Lemma \ref{compexp}. 

Finally, these two pieces let us compute $z' = r^{E}\exp(i\theta')$ to within precision $2^{-\poly(n)}$, as desired.
\end{proof}

\end{appendix}

\end{document}